\documentclass[copyright,creativecommons]{eptcs}
\providecommand{\event}{EXPRESS/SOS 2012}
\usepackage{breakurl}
\usepackage{latexsym}                   
\usepackage{stmaryrd}                   
\usepackage{graphicx}                   
\makeatletter
\newif\if@qeded
\global\@qededfalse
\def\qed{\hfill$\Box$\global\@qededtrue}
\def\qedif{\if@qeded\else\qed\fi\global\@qededfalse}
\makeatother
\newtheorem{defi}{Definition}
\newtheorem{theo}{Theorem}
\newtheorem{prop}{Proposition}
\newtheorem{lemm}{Lemma}
\newtheorem{coro}{Corollary}
\newtheorem{exam}{Example}
\newtheorem{obse}{Observation}
\newtheorem{post}{Postulate}

\newenvironment{definition}[1]{\begin{defi} \rm \label{df:#1} }{\end{defi}}
\newenvironment{theorem}[1]{\begin{theo} \rm \label{thm:#1} }{\end{theo}}
\newenvironment{proposition}[1]{\begin{prop} \rm \label{pr:#1} }{\end{prop}}

\newenvironment{corollary}[1]{\begin{coro} \rm \label{cor:#1} }{\end{coro}}
\newenvironment{example}[1]{\begin{exam} \rm \label{ex:#1} }{\end{exam}}
\newenvironment{observation}[1]{\begin{obse} \rm \label{obs:#1} }{\end{obse}}
\newenvironment{postulate}[1]{\begin{post} \rm \label{post:#1} }{\end{post}}
\newenvironment{proof}{\begin{trivlist} \item[\hspace{\labelsep}\bf Proof:]}
               {\qedif\end{trivlist}}

\newcommand{\df}[1]{Definition~\ref{df:#1}}
\newcommand{\thm}[1]{Theorem~\ref{thm:#1}}
\newcommand{\pr}[1]{Proposition~\ref{pr:#1}}

\newcommand{\cor}[1]{Corollary~\ref{cor:#1}}

\newcommand{\obs}[1]{Observation~\ref{obs:#1}}
\newcommand{\sect}[1]{Section~\ref{sec:#1}}
\newcommand{\tab}[1]{Table~\ref{tab:#1}}
\newfont{\bbb}{bbm10 scaled 1100}        
\newfont{\bbbs}{bbm10 scaled 900}        
\newcommand{\denote}[1]{\mbox{\bbb [}#1\mbox{\bbb ]}} 
\newcommand{\IN}{\mbox{\bbb N}}          
\newcommand{\IR}{\mbox{\bbb R}}          
\newcommand{\IT}{\mbox{\bbb T}}          
\newcommand{\T}{{\rm T}}                 
\newcommand{\D}{{\cal D}}                
\newcommand{\fL}{{\cal L}}               
\newcommand{\fT}{{\cal T}}               
\newcommand{\bC}{{\bf C}}                
\newcommand{\bD}{{\bf D}}                
\newcommand{\bE}{{\bf E}}                
\newcommand{\bU}{{\bf U}}                
\newcommand{\bV}{{\bf V}}                
\newcommand{\bW}{{\bf W}}                
\newcommand{\bZ}{{\bf Z}}                
\newcommand{\E}{E}                       
\newcommand{\p}{P}                       
\newcommand{\s}{s}                       
\newcommand{\B}{{\cal B}}                
\newcommand{\V}{\mathcal{X}}             
\newcommand{\stp}{\mbox{\sc stop}}       
\newcommand{\dv}{\mbox{\sc div}}         
\newcommand{\obox}{\mathbin\Box}         
\newcommand{\conceal}{/}                 
\newcommand{\plat}[1]{\raisebox{0pt}[0pt][0pt]{#1}}     
\newcommand{\dcup}{\stackrel{\mbox{\huge .}}{\cup}}	
\newcommand{\goto}[1]{\mathrel{\stackrel{#1}{\raisebox{0pt}
	[3pt][0pt]{$\longrightarrow$}}}}	        
\newcommand{\gonotto}[1]{\hspace{4pt}\not\hspace{-4pt}	
	\stackrel{#1\ }{\longrightarrow}}
\newcommand{\dto}[1]{\mathrel{\stackrel{#1}{\raisebox{.7pt} 
	[3pt][0pt]{$\scriptstyle=\hspace{-2pt}\Longrightarrow$}}}}
\newcommand{\dom}{{\it dom}}                            
\newcommand{\bis}[2][]{		             		
	\mathrel{\raisebox{.3ex}{$\;\underline{\makebox[.7em]{$\leftrightarrow$}}$}
                  \,_{#2}^{\,#1}}}
\newcommand{\eqa}{\mathrel{\plat{$\stackrel{\alpha}=$}}} 
\def\titlerunning{Musings on Encodings and Expressiveness}
\title{\titlerunning}
\author{Rob van Glabbeek
\institute{NICTA, Sydney, Australia}
\institute{School of Computer Science and Engineering,
University of New South Wales, Sydney, Australia}
\email{rvg@cs.stanford.edu}
}
\def\authorrunning{R.J. van Glabbeek}
\begin{document}
\maketitle

\begin{abstract}
This paper proposes a definition of what it means for one system
description language to encode another one, thereby enabling an ordering
of system description languages with respect to expressive power.
I compare the proposed definition with other definitions of encoding and
expressiveness found in the literature, and illustrate it on a case study:
comparing the expressive power of CCS and CSP.
\end{abstract}

\section{Introduction}

This paper aims at answering the question what it means for one
language to encode another one, and make this definition applicable to
order system description languages like CCS, CSP and the
$\pi$-calculus with respect to their expressive power.

To this end it proposes a unifying concept of correct translation between two languages,
and adapts it to translations \emph{up to} a semantic equivalence, for languages with a
denotational semantics that interprets the operators and recursion constructs as operations
on a set of values, called a \emph{domain}.  Languages can be partially ordered by their
expressiveness up to the chosen equivalence according to the existence of correct
translations between them.

The concept of a [correct] translation between system description languages (or \emph{process
calculi}) was first formally defined by Boudol \cite{Bo85}. There, and in most other
related work in this area, the domain in which a system description language is
interpreted consists of the closed expressions from the language itself. In \cite{vG94a} I
have reformulated Boudol's definition, while dropping the requirement that the domain of
interpretation is the set of closed terms. This allows (but does not
enforce) a clear separation of syntax and semantics, in the tradition of universal
algebra.  Nevertheless, the definition employed in \cite{vG94a} only deals with the case
that all (relevant) elements in the domain are denotable as the interpretations of closed
terms. Examples~\ref{ex:numbers} and~\ref{ex:undenotable} herein will present
situations where such a restriction is undesirable.
In addition, both \cite{Bo85} and \cite{vG94a} require the semantic equivalence $\sim$ under
which two languages are compared to be a congruence for both of them.
This is too severe a restriction to capture some recent encodings.

The current paper aims to generalise the concept of a correct translation as much as possible, so
that it is uniformly applicable in many situations, and not just in the world of process
calculi. Also, it needs to be equally applicable to encodability and separation results,
the latter saying that an encoding of one language in another does not exists.
At the same time, it tries to derive this concept from a unifying principle,
rather than collecting a set of criteria that justify a number of known
encodability and separation results that are intuitively justified.

In Sections~\ref{sec:up to} and~\ref{sec:respects} I propose in fact two notions of
encoding: \emph{correct} and \emph{valid} translations up to $\sim$. The former drops the
restriction on denotability and $\sim$ being a congruence for the whole target language,
but it requires $\sim$ to be a congruence for the source language, as well as the source's
image within the target. The latter drops both congruence requirements, but at the expense
of requiring denotability by closed terms.  In situations where $\sim$ is a congruence for
the source language's image within the target language \emph{and} all semantic values are
denotable, the two notions agree.  \advance\textheight 13.6pt

\section{Correct translations and expressiveness}

A language consists of \emph{syntax} and \emph{semantics}.
The syntax determines the valid expressions in the language.
The semantics is given by a mapping $\denote{\ \ }$ that associates
with each valid expression its meaning, which can for instance be an
object, concept or statement.
This mapping determines the set $\D$ of all objects, concepts or
statements that can be denoted in the language, namely as its image.

A correct translation of one language into another is a mapping
from the valid expressions in the first language to those in the
second, that preserves their meaning, i.e.\ such that the meaning of
the translation of an expression is the same as the meaning of the
expression being translated. In order to formalise this, I represent a
language $\fL$ as a pair $(\IT_{\fL},\denote{\ \ }_\fL)$ of a set
$\IT_{\fL}$ of valid expressions in $\fL$ and a surjective mapping
$\denote{\ \ }_\fL:\IT_\fL\rightarrow \D_\fL$ from $\IT_\fL$ in some
set of meanings $\D_\fL$.

\begin{definition}{translation}
A \emph{translation} from a language $\fL$ into a language $\fL'$ is a
mapping $\fT: \IT_\fL \rightarrow \IT_{\fL'}$. It is \emph{correct} when
$\denote{\fT(\E)}_{\fL'} = \denote{\E}_\fL$ for all $\E\in \IT_\fL$.
Language $\fL'$ is at least as \emph{expressive} as $\fL$ if a correct
translation exists.
\end{definition}
\begin{figure}[h]
\vspace{-13.6pt}
\begin{center}
\begin{picture}(6,3.5)(0,.5)
\put(0,3.5){\makebox[0pt]{\textit{dog}}}
\put(1,3.6){\vector(1,0){4}}
\put(3,3.7){$\fT$}
\put(-.4,2.4){$\denote{\ \ }_{\it English}$}
\put(5.2,2.4){$\denote{\ \ }_{\it French}$}
\put(.7,3){\vector(1,-1){.8}}
\put(5.3,3){\vector(-1,-1){.8}}
\put(6,3.5){\makebox[0pt]{\textit{chien}}}
\put(3,0){\makebox[0pt]{\includegraphics[width=3cm]{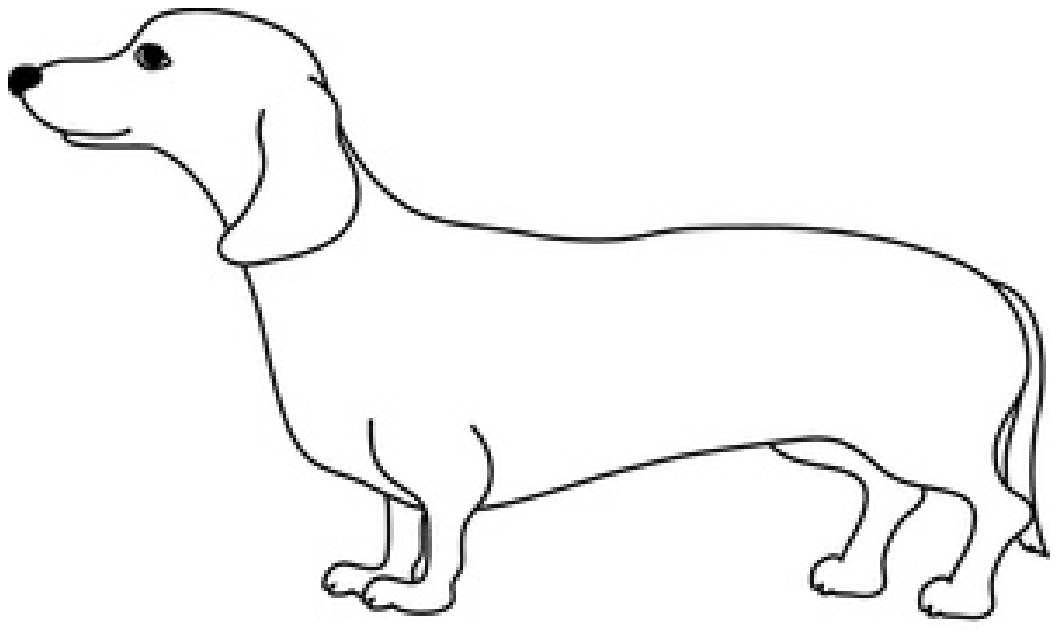}}}
\end{picture}
\end{center}
\caption{The essence of a correct translation}
\label{dog}
\end{figure}
This fundamental notion is illustrated in Figure~\ref{dog}.
It is not hard to see that a correct translation from $\fL$ to $\fL'$ exists
if and only if anything that can be expressed in $\fL$ can also be
expressed in $\fL'$, i.e.\ iff $\D_\fL \subseteq \D_{\fL'}$.%

In this paper I will argue that this simple notion of a correct translation,
when instantiated with appropriate proposals for $\denote{\ \ }$ and $\D$,
is a suitable definition of an encoding from one system description language
into another, and thereby a suitable basis for classifying such languages
w.r.t.\ expressiveness. 

\section{Dividing out a semantic equivalence}\label{sec:dividing out}

\begin{definition}{process graph}
A \emph{process graph} over an alphabet $Act$ is a triple
$(S,I,\rightarrow)$ with $S$ a set of \emph{states}, $I\in S$ the
\emph{initial state}, and $\mathord{\rightarrow} \subseteq S\times Act
\times S$ the \emph{transition relation}.
\end{definition}
In other words, a process graph is a labelled transition system
equipped with an initial state.

One way to apply the above definition of a translation to system
description languages like CCS and CSP would be to take variable-free
(and hence recursion-free) versions of those languages, and to define
the meaning $\denote{\p}$ of a CCS or CSP expression $\p$ to be the
process graph $G_P:=(S,\p,\rightarrow)$ with as set of states $S$ the set of
all CCS/CSP expressions, as initial state the expression $\p$, and
$\rightarrow$ being the transition relation generated by the standard
structural operational semantics of these languages. A variant of this
idea is to reduce $S$ to the states that are \emph{reachable} from $\p$
by following transitions.

Now it happens to be case that the reachable part of each process
graph that can be denoted by a CSP expression is \emph{isomorphic},
but in general not \emph{equal}, to one that can be denoted by a CCS
expression. As an example consider the CCS and CSP constants for
\emph{inaction}. In CCS this constant is called $0$
whereas in CSP it is called $\stp$. The operational semantics
generates no outgoing transitions of either process. It is therefore
tempting to translate the CSP constant $\stp$ into the CCS constant $0$.
Yet, this is not a correct translation in the current set-up, as
the process graph with initial state $0$ and no other states or
transitions is different from the one with initial state $\stp$.

One way to deal with this anomaly is to relax \df{translation} by
defining an appropriate semantic equivalence $\sim$ on $\D_\fL \cup
\D_{\fL'}$ and merely requiring that the meanings of an expression
and its translation are \emph{equivalent}.

\begin{definition}{translation up to}
A translation $\fT: \IT_\fL \rightarrow \IT_{\fL'}$ from a language $\fL$ into a language $\fL'$
is \emph{correct up to} a semantic equivalence $\sim$ on $\D_\fL \cup\D_{\fL'}$ when
$\denote{\fT(\E)}_{\fL'} \sim \denote{\E}_\fL$ for all $\E\in \IT_\fL$.
\end{definition}
In the example above, an appropriate candidate for $\sim$ could be
isomorphism of reachable parts.

In some sense, introducing an appropriate semantic equivalence $\sim$,
or maybe a preorder, appears to be the only reasonable way to allow
intuitively correct translations, such as of 0 by \stp. Nevertheless,
it need not be seen as a relaxation---and hence abandonment---of
\df{translation}, but rather as an appropriate instantiation.
Namely the meaning of a CCS or CSP expression $\p$ is no
longer a process graph $G$, but instead the equivalence class $[G]_\sim$
of all process graphs in $\D_{\rm CCS} \cup \D_{\rm CSP}$ that are
equivalent to $G$.

\begin{observation}{translation up to}
Let $\fL=(\IT_\fL,\denote{\ \ }_\fL)$ and
$\fL'=(\IT_{\fL'},\denote{\ \ }_{\fL'})$ be two languages, and $\fT: \IT_\fL \rightarrow \IT_{\fL'}$
a correct translation between them up to an equivalence $\sim$ on $\D_\fL \cup\D_{\fL'}$.
Then $\fT$ is a correct translation between the languages
$(\IT_\fL,\denote{\ \ }_\fL^\sim)$ and $(\IT_{\fL'},\denote{\ \ }_{\fL'}^\sim)$,
where $\denote{\E}_\fL^\sim$ is defined to be $[\denote{\E}_\fL]_\sim$.
\end{observation}
Hence, correct translations up to some equivalence can be seen as special
cases of correct translations. In doing so, it may appear problematic that the
meaning $\denote{\E}_\fL^\sim$ of an expression $\E\in\IT_\fL$ becomes
dependent on the semantic domain $\D_{\fL'}$ of the other language,
namely by $\denote{\E}_\fL^\sim$ being the class of all processes in
$\D_\fL \cup\D_{\fL'}$ that are equivalent with $\denote{\E}_\fL$. This
worry can be alleviated by using, instead of $\D_\fL \cup\D_{\fL'}$,
a natural class of which both $\D_\fL$ and $\D_{\fL'}$ are subsets.
In the example above this could for instance be the class of all
process graphs (over a suitable alphabet).

\section{Translating operators}\label{sec:operators}

Up to isomorphism of reachable parts, so certainly up to coarser
equivalences such as strong bisimilarity, the variable-free fragments
of CSP and CCS with finitary choice are equally expressive. Namely each of them can express
exactly the (equivalence classes of) finite process graphs. Here a
process graph is finite if it has finitely many states and
transitions, and no loops. In fact, these languages do not lose any
expressiveness when omitting their parallel compositions, for parallel
composition is not needed to denote any finite process graph.

Hence the treatment above does not address the question whether one of
the \emph{operators} of one language, such as parallel composition,
can be mimicked by an operator or combination of operators in the
other.  This is to be blamed on the absence of variables. Once we
admit variables in the language, the CCS parallel composition
corresponds to the CCS expression $X|Y$, where $X$ and $Y$ are
process variables, and a correct translation to CSP ought to translate
this expression to a valid CSP expression---a CSP context built from
CSP operators and the variables $X$ and $Y$.

Henceforth, I consider single-sorted languages $\fL$ in which
\emph{expressions} or \emph{terms} are built from variables (taken
from a set $\V$) by means of operators (including constants) and
possibly recursion constructs.\footnote{In \sect{compositionality}
two postulates will be presented that restrict the class of languages
considered in this paper.}
The semantics of such a language is given by a domain of
values $\bD$, and an interpretation of each $n$-ary operator $f$ of
$\fL$ as an $n$-ary operation $f^\bD: \bD^n\rightarrow \bD$ on $\bD$.
Using the equations $$\denote{X}_\fL(\rho) = \rho(X) \qquad \mbox{and} \qquad
\denote{f(E_1,\ldots,E_n)}_\fL(\rho) = f^\bD(\denote{E_1}_\fL(\rho),\ldots,\denote{E_n}_\fL(\rho))$$
this allows an inductive definition of the meaning $\denote{\E}_\fL$
of an $\fL$-expression $\E$ as a function of type
$(\V\!\!\rightarrow\bD)\rightarrow\bD$, associating a value
$\denote{E}_\fL(\rho) \mathbin\in \bD$ to $E$ that depends on the choice of a
\emph{valuation} $\rho\!:\V\!\!\!\rightarrow\!\bD$. The valuation associates a
value from $\bD$ with each variable.  Moreover, $\denote{E}_\fL(\rho)$
only depends on the restriction of $\rho$ to those variables that
occur free in $\E$. In this setting, the class $\D_\fL$ of possible
meanings of $\fL$-expressions is a subclass of
$(\V\!\!\rightarrow\bD)\rightarrow\bD$.  Hence, a translation $\fT: \IT_\fL \rightarrow \IT_{\fL'}$
between two such languages $\fL$ and $\fL'$ that employ the same set $\V$ of
variables and are interpreted in the same domain $\bD$ is correct
when $\denote{\fT(\E)}_{\fL'}(\rho) = \denote{\E}_\fL(\rho)$ 
for all $\E\in \IT_\fL$ and all valuations $\rho:\V\rightarrow\bD$.

Since normally the names of variables are irrelevant and the
cardinality of the set of variables satisfies only the requirement
that it is ``sufficiently large'', no generality is lost by insisting
that two (system description) languages whose expressiveness is being
compared employ the same set of (process) variables. On the other
hand, two languages $\fL$ and $\fL'$ may be interpreted in different
domains of values $\bD$ and $\bD'$. Without dividing out a semantic
equivalence, one must insist that $\bD\subseteq \bD'$;
otherwise no correct translation from $\fL$ into $\fL'$ exists.  When
$\bD\subseteq \bD'$ also $(\V\rightarrow\bD) \subseteq(\V\rightarrow\bD')$,
so any function $(\V\rightarrow\bD')\rightarrow\bD'$ restricts to a function
$(\V\rightarrow\bD)\rightarrow\bD'$. For the purpose of comparing the
expressive power of $\fL$ and $\fL'$, the semantics of $\fL'$ can be
taken to be the mapping $\denote{\ \ }_{\fL'}:\IT_{\fL'} \rightarrow
((\V\rightarrow\bD)\rightarrow\bD')$, where $\denote{E}_{\fL'}(\rho)$
with $E\in\IT_{\fL'}$ is considered for valuations $\rho:\V\rightarrow\bD$ only.
This restriction entails that when translating $\fL$ into $\fL'$ I
compare the meaning of $\fL$-expressions and their translations only
under valuations within the domain $\bD$ in which $\fL$ is interpreted.
A translation $\fT: \IT_\fL \rightarrow \IT_{\fL'}$ from $\fL$ to $\fL'$
remains correct when $\denote{\fT(\E)}_{\fL'}(\rho) = \denote{\E}_\fL(\rho)$ 
for all $\E\in \IT_\fL$ and all valuations $\rho:\V\rightarrow\bD$.

\begin{example}{numbers}
Let $\fL$ be the language whose syntax consists of a binary operator $+$,
interpreted as addition in the domain $\IN$ of the natural numbers.
So $\IT_\fL$ contains expressions such as $X+(Y+Z)$.
$\fL'$ is the language with unary operators $e^x$ and $\ln(x)$, interpreted as exponentiation
and the natural logarithm on the reals $\IR$, as well as the binary operator $\times$
of multiplication. If you do not like partial functions, the domain $\IR$ can be extended with a
special value $\bot$ to capture undefined outcomes. Note that $\IN\subset\IR$.
Using that $\ln(e^x)=x$, the $\fL$-expression $X+Y$ can be translated into the
$\fL'$-expression $\ln(e^X\times e^Y)$. Using this, a translation $\fT:\IT_\fL \rightarrow\IT_{\fL'}$
is defined inductively by $\fT(X):=X$ and $\fT(E+F):=\ln(e^{\fT(E)}\times e^{\fT(E)})$.
\end{example}

\section{Correct translations up to a congruence}\label{sec:up to}

This section aims at integrating the instantiations of the notion of a
correct translation proposed in Sections~\ref{sec:dividing out} and~\ref{sec:operators}.
Let $\fL$ and $\fL'$ be two languages of the type considered in
\sect{operators}, with semantic mappings 
$\denote{\ \ }_{\fL}:\IT_{\fL} \rightarrow ((\V\rightarrow\bV)\rightarrow\bV)$
and
$\denote{\ \ }_{\fL'}:\IT_{\fL'} \rightarrow ((\V\rightarrow\bV')\rightarrow\bV')$.
Here $\bV$ and $\bV'$ are domains of interpretation prior to
quotienting by an appropriate semantic equivalence; they might be sets of
process graphs with as states closed CCS expressions and closed CSP
expressions, respectively. In order to compare these languages
w.r.t.\ their expressive power I need a semantic equivalence $\sim$ that is defined
on a unifying domain of interpretation $\bZ$, with $\bV,\bV' \subseteq\bZ$.
Let $\bU:=\{v\in\bV'\mid \exists v\in \bV.~v'\sim v\}$.

\begin{definition}{equivalence of valuations}
Two valuations $\eta,\rho:\V\rightarrow\bZ$ are \emph{$\sim$-equivalent},
$\eta\sim\rho$, if $\eta(X)\sim\rho(X)$ for each $X\in \V$.
\end{definition}
In case there exists a $v\in \bV$ for which there is no
$\sim$-equivalent $v'\in \bV'$, there is no correct translation from
$\fL$ into $\fL'$ up to $\sim$.  Namely, the semantics of $\fL$
describes, among others, how any $\fL$-operator evaluates the argument
value $v$, and this aspect of the language has no counterpart in $\fL'$.
Therefore, I will require\vspace{-1.5ex}
\begin{equation}\label{related}
\forall v\in \bV.~ \exists v'\in \bV'.~ v'\sim v.
\end{equation}
This implies that for any valuation $\rho:\V\rightarrow\bV$ there is a
valuation $\eta:\V\rightarrow\bV'$ with $\eta\sim\rho$.

\begin{definition}{correct translation}
A translation $\fT$ from $\fL$ into $\fL'$ is \emph{correct up to $\sim$} iff
(\ref{related}) holds and\\ $\denote{\fT(\E)}_{\fL'}(\eta) \sim \denote{\E}_\fL(\rho)$ for all
$\E\in \IT_\fL$ and all valuations $\eta:\V\rightarrow\bV'$ and $\rho:\V\rightarrow\bV$ with
$\eta\sim\rho$.
\end{definition}
Note that a correct translation as defined in \sect{operators} is exactly a correct
translation up to the identity relation.
If a correct translation up to $\sim$ from $\fL$ into $\fL'$ exists, then $\sim$ must be a
congruence for $\fL$.

\begin{definition}{congruence}
An equivalence relation $\sim$ is a \emph{congruence} for a language
$\fL$ interpreted in  a semantic domain $\bV$ if
$\denote{E}_\fL(\nu)\sim\denote{E}_\fL(\rho)$
for any $\fL$-expression $E$ and any valuations
$\nu,\rho:\V\rightarrow \bV$ with $\nu \sim \rho$.
\end{definition}

\begin{proposition}{congruence}
If a correct translation up to $\sim$ from $\fL$ into $\fL'$ exists, then $\sim$ is a
congruence for $\fL$.
\end{proposition}

\begin{proof}
Let $\fT$ be a correct translation up to $\sim$ from $\fL$ into $\fL'$.
Let $E\in\IT_\fL$ and let $\nu,\rho:\V\!\rightarrow\bV$ with $\nu\mathord\sim\rho$.
By (\ref{related}) there is a valuation $\eta\!:\!\V\!\!\rightarrow\!\bV'$ with $\eta\mathbin\sim\nu$.
Hence $\denote{\E}_{\fL}(\nu) \mathbin\sim \denote{\fT(\E)}_{\fL'}(\eta) \mathbin\sim \denote{\E}_\fL(\rho)$.%
\end{proof}
The existence of a correct translation up to $\sim$ from $\fL$ into $\fL'$ does not imply
that $\sim$ is a congruence for $\fL'$. However, $\sim$ has the properties of a congruence for
those expressions of $\fL'$ that arise as translations of expressions of $\fL$, when
restricting attention to valuations into $\bU$. I call this a \emph{congruence for} $\fT(\fL)$.%
\begin{definition}{weak congruence}
Let $\fT: \IT_\fL \rightarrow \IT_{\fL'}$ be a translation from $\fL$ into $\fL'$.
An equivalence $\sim$ on $\IT_{\fL'}$ is a \emph{congruence for} $\fT(\fL)$
if $\denote{\fT(E)}_{\fL'}(\nu)\mathbin\sim\denote{\fT(E)}_{\fL'}(\eta)$
for any $E\mathbin\in\IT_\fL$ and $\nu,\eta\!:\!\V\!\!\!\rightarrow \!\bU$ with $\nu \mathbin\sim \eta$.%
\end{definition}

\begin{proposition}{weak congruence}
If a correct translation up to $\sim$ from $\fL$ into $\fL'$ exists, then $\sim$ is a
congruence for $\fT(\fL)$.
\end{proposition}

\begin{proof}
Let $\fT$ be correct up to $\sim$ from $\fL$ into $\fL'$.
Let $E\in\IT_\fL$ and let $\nu,\eta:\V\rightarrow\bU$ with $\nu\sim\eta$.
By definition of $\bU$ there is a $\rho:\V\rightarrow\bV$ with $\rho\sim\nu$.
Hence $\denote{\fT(\E)}_{\fL'}(\nu) \sim \denote{\E}_\fL(\rho) \sim\denote{\fT(\E)}_{\fL'}(\eta)$.
\end{proof}
In the rest of this section I will show how the concept of a correct transition up to
$\sim$ can be seen as an instantiation of the notion of correct translation, analogously
to the situation in \sect{dividing out}. To this end I need to unify the types of the
semantic mappings $\denote{\ \ }_{\fL}$ and $\denote{\ \ }_{\fL'}$, say as
$\denote{\ \ }_{\fL}:\IT_{\fL} \rightarrow ((\V\rightarrow\bE)\rightarrow\bD)$
and
$\denote{\ \ }_{\fL'}:\IT_{\fL'} \rightarrow ((\V\rightarrow\bE)\rightarrow\bD)$.%
\footnote{In fact, it suffices to obtain mappings
  $\denote{\ \ }_{\fL}:\IT_{\fL} \rightarrow ((\V\rightarrow\bE)\rightarrow\bD)$
  and
  $\denote{\ \ }_{\fL'}:\IT_{\fL'} \rightarrow ((\V\rightarrow\bE')\rightarrow\bD')$
  satisfying $((\V\rightarrow\bE)\rightarrow\bD) \subseteq ((\V\rightarrow\bE')\rightarrow\bD')$,
  and hence $\bE'=\bE$ and $\bD\subseteq\bD'$. However, any
  mapping $\denote{\ \ }_{\fL}:\IT_{\fL} \rightarrow ((\V\rightarrow\bE)\rightarrow\bD)$
  is also a mapping
  $\denote{\ \ }_{\fL}:\IT_{\fL} \rightarrow ((\V\rightarrow\bE)\rightarrow\bD')$,
  so one can just as well use $\bD'$ for $\bD$.}
This unification process involves dividing out the semantic equivalence $\sim$, as well as changing
the type of a semantic mapping without tampering with the essence of its meaning.
Below I propose two methods for doing so. The first method applies when $\sim$
is a congruence for both $\fL$ and $\fL'$, whereas the second merely
requires that it is a congruence for $\fL$.
In both cases, the semantic mappings $\denote{\ \ }_{\fL}$ and
$\denote{\ \ }_{\fL'}$ can be understood to be of types
$\IT_{\fL} \rightarrow ((\V\rightarrow\bV)\rightarrow\bZ)$ and
$\IT_{\fL'} \rightarrow ((\V\rightarrow\bV')\rightarrow\bZ)$, respectively.
Dividing out $\sim$ yields the quotient domain $\bD:=\bZ/\!\sim:=\{[z]_\sim\mid z\in \bZ\}$,
consisting of the $\sim$-equivalence classes of elements of $\bZ$,
together with the mappings
$\denote{\ \ }^\sim_\fL:\IT_{\fL} \rightarrow ((\V\rightarrow\bV)\rightarrow\bD)$ and
\plat{$\denote{\ \ }^\sim_\fL:\IT_{\fL'} \rightarrow ((\V\rightarrow\bV')\rightarrow\bD)$},
where $\denote{E}_\fL^{\sim}(\rho) := [\denote{E}_\fL(\rho)]_\sim$.

\subsection{Translations up to a congruence for both languages}
\label{sec:A}

Let $\sim$ be a congruence for both $\fL$ and $\fL'$.
Take $\bW := \{v'' \in \bZ \mid \exists v\in \bV.~ v\sim v''\}$
and likewise $\bW' := \{v'' \in \bZ \mid \exists v'\in \bV'.~ v'\sim v''\}$.
Furthermore, $\bC:=\bW/\!_\sim$ and $\bC':=\bW'/\!_\sim$.
By (\ref{related}), $\bW\subseteq\bW'$ and $\bC\subseteq\bC'\subseteq\bD$.

Now $\denote{\ \ }_\fL^\sim$ can be recast as a function of
type $\IT_{\fL} \rightarrow ((\V\rightarrow\bC)\rightarrow\bD)$;
namely by defining $\denote{E}_{\fL}^{\sim}(\theta)$ with
$\theta:\V\rightarrow\bC$ to be $\denote{E}_{\fL}^{\sim}(\rho)$, for
any valuation $\rho:\V\rightarrow\bV$ such that
$\theta(X)=[\rho(X)]_\sim$ for all $X\in \V$.
The congruence property of $\sim$ ensures that the value
\plat{$\denote{E}_{\fL}^{\sim}(\theta)\in\bD$} is independent of the
choice of the representatives $\rho(X)$ in the equivalence classes $\theta(X)$.

Likewise, $\denote{\ \ }_{\fL'}^\sim$ can be recast as a
function of type $\IT_{\fL'} \rightarrow((\V\rightarrow\bC')\rightarrow\bD)$,
which, as in \sect{operators}, can be restricted to a function of type
$\IT_{\fL'} \rightarrow((\V\rightarrow\bC)\rightarrow\bD)$.
A translation $\fT: \IT_\fL \rightarrow \IT_{\fL'}$ from $\fL$ into $\fL'$
can be defined to be \emph{correct up to} $\sim$ when (\ref{related}) holds and
$\denote{\fT(\E)}_{\fL'}^{\sim}(\theta) = \denote{\E}_\fL^{\sim}(\theta)$ for all
$\E\in \IT_\fL$ and all valuations $\theta:\V\rightarrow\bC$.
It is not hard to check that this definition agrees with \df{correct translation}.

\subsection{Translations up to a congruence for the source language}
\label{sec:B}

Let $\sim$ be a congruence for $\fL$.
Recast $\denote{\ \ }_{\fL}^\sim$ as a function of type
$\IT_{\fL} \rightarrow ((\V\rightarrow\bU)\rightarrow\bD)$
by defining $\denote{E}_{\fL}^{\sim}(\eta)$ with
$\eta:\V\rightarrow\bU$ to be $\denote{E}_{\fL}^{\sim}(\rho)$, for
any valuation $\rho:\V\rightarrow\bV$ with $\rho\sim\eta$.
The congruence property of $\sim$ ensures that the value
$\denote{E}_{\fL}^{\sim}(\eta)\in\bD$ is independent of the choice of
the representative valuation $\rho$.

Since $\bU\subseteq \bV$ also $(\V\rightarrow\bU) \subseteq
(\V\rightarrow\bV)$, and therefore any function
$(\V\rightarrow\bV)\rightarrow\bD$ restricts to a function
$(\V\rightarrow\bU)\rightarrow\bD$.
This way, $\denote{\ \ }_{\fL'}^\sim$ can be recast as
a function of type $\IT_{\fL'} \rightarrow
((\V\rightarrow\bU)\rightarrow\bD)$ as well, and unification is achieved.
Now a translation $\fT: \IT_\fL \rightarrow \IT_{\fL'}$ from $\fL$ into $\fL'$
can be defined to be \emph{correct up to} $\sim$ when (\ref{related}) holds and
$\denote{\fT(\E)}_{\fL'}^{\sim}(\eta) = \denote{\E}_\fL^{\sim}(\eta)$ for all
$\E\in \IT_\fL$ and all valuations $\eta:\V\rightarrow\bU$.
It is straightforward that this definition agrees with \df{correct translation}.

\section{A hierarchy of expressiveness preorders}

An equivalence $\sim$ on a class $\bZ$ is said to be \emph{finer},
\emph{stronger}, or \emph{more discriminating} than another
equivalence $\approx$ on $\bZ$ if $p \sim q \Rightarrow p\approx q$
for all $p,q\in\bZ$. 
\begin{theorem}{hierarchy}
Let $\fT: \IT_\fL \rightarrow \IT_{\fL'}$ be a translation from $\fL$ into $\fL'$, and let
$\sim,\approx$ be congruences for $\fT(\fL)$, with $\sim$ finer than $\approx$.
If $\fT$ is correct up to $\sim$, then it is also correct up to $\approx$.
\end{theorem}

\begin{proof}
Let $\bU^\approx := \{v' \in \bV' \mid \exists v\in \bV.~ v\approx v''\}$.
Let $\fT$ be correct up to $\sim$.
Then $\denote{\fT(\E)}_{\fL'}(\eta) \sim \denote{\E}_\fL(\rho)$ for all
$\E\in \IT_\fL$ and all $\eta:\V\!\!\rightarrow\bV'$ and $\rho:\V\!\!\rightarrow\bV$
with $\eta\sim\rho$.
To establish that $\fT$ also is correct up to $\approx$, let $E\in\IT_\fL$,
$\nu:\V\!\!\rightarrow\bV'$ and $\rho:\V\!\!\rightarrow\bV$ with $\nu\approx\rho$.
Take $\eta:\V\rightarrow\bV'$ with $\eta\sim\rho$---it exists by (\ref{related}).
Then $\denote{\fT(\E)}_{\fL'}(\eta) \sim \denote{\E}_\fL(\rho)$ and hence
$\denote{\fT(\E)}_{\fL'}(\eta) \approx \denote{\E}_\fL(\rho)$.
By (\ref{related}) both $\eta$ and $\nu$ are of type
$\V\!\!\!\rightarrow\bU^\approx$.
Since $\approx$ is a congruence for $\fT(\fL)$ and $\nu\mathbin\approx\eta$,
$\denote{\fT(\E)}_{\fL'}(\nu) \approx \denote{\fT(\E)}_\fL(\eta) \approx \denote{\E}_{\fL'}(\rho)$.%
\end{proof}
When it is necessary to divide out a semantic equivalence, the quality
of a translation depends on the choice of this equivalence.  In no way
would I want to suggest that a language $\fL'$ is at least as
expressive as $\fL$ when there is a correct translation of $\fL$ up to
\emph{some} equivalence---the equivalence does \emph{not} appear in
the scope of an existential quantifier. In fact, this would make any
two languages equally expressive, namely by using the universal
equivalence, relating any two processes.
Instead, the equivalence needs to be chosen carefully to match the
intended applications of the languages under comparison. In general,
as show by \thm{hierarchy}, using a finer equivalence yields a
stronger claim that one language can be encoded in another.
On the other hand, when separating two languages $\fL$ and $\fL'$ by
showing that $\fL$ \emph{cannot} be encoded in $\fL'$, a coarser
equivalence generally yields a stronger claim.

The following corollary of \thm{hierarchy} is a powerful tool for
proving the nonexistence of translations.%

\begin{corollary}{congruence}
If there is a correct translation up to $\sim$ from $\fL$ into $\fL'$,
and $\approx$ is a congruence for $\fL'$ that is coarser than $\sim$,
then $\approx$ is a congruence for $\fL$.
\end{corollary}

\begin{proof}
By combining \thm{hierarchy} and \pr{congruence}.
\end{proof}

\begin{proposition}{identity}
If $\sim$ is a congruence for a language $\fL$, then the identity is a
correct translation up to $\sim$ from $\fL$ into itself.
\end{proposition}
\begin{proof}
Immediately from Definitions~\ref{df:correct translation} and~\ref{df:congruence}.
\end{proof}

\begin{theorem}{composition}
If correct translations up to $\sim$ exists from $\fL_1$ into $\fL_2$ and from $\fL_2$
into $\fL_3$, then there is a correct translation up to $\sim$ from $\fL_1$ into $\fL_3$.
\end{theorem}

\begin{proof}
For $i=1,2,3$ let $\denote{\ \ }_{\fL_i}:\IT_{\fL_i} \rightarrow
((\V\rightarrow\bV_i)\rightarrow\bV_i)$, and for $k=1,2$ let
$\fT_k:\IT_{\fL_k}\rightarrow\IT_{\fL_{k+1}}$ be correct translations up to
$\sim$ from $\fL_k$ to $\fL_{k+1}$. I will show that the translation
$\fT_2\circ\fT_1:\IT_{\fL_1}\rightarrow\IT_{\fL_3}$ from $\fL_1$ to $\fL_3$, given by
$\fT_2\circ\fT_1(E)=\fT_2(\fT_1(E))$, is a correct up to $\sim$.

By assumption, $\denote{\fT_1(\E)}_{\fL_2}(\eta) \sim \denote{\E}_{\fL_1}(\rho)$ for
all $\E\in \IT_{\fL_1}$ and all $\eta:\V\rightarrow\bV_2$ and
$\rho:\V\!\rightarrow\bV_1$ with $\eta\sim\rho$, and likewise
$\denote{\fT_2(F)}_{\fL_3}(\nu) \sim \denote{F}_{\fL_2}(\eta)$ for
all $F\in \IT_{\fL_2}$ and all $\nu:\V\!\rightarrow\bV_3$ and
$\eta:\V\!\rightarrow\bV_2$ with $\nu\sim\eta$.
Let $\E\in \IT_{\fL_1}$, $\nu:\V\!\rightarrow\bV_3$ and $\rho:\V\!\rightarrow\bV$
with $\nu\sim \rho$; I need to
show that \plat{$\denote{\fT_2\circ\fT_1(\E)}_{\fL_3}(\nu) \sim \denote{\E}_{\fL_1}(\rho)$}.

Let $\eta:\V\rightarrow\bV_2$ be a valuation with $\eta\sim\rho$---it exists by (\ref{related}).
Then $\nu\sim\eta$. Taking $F:=\fT_1(E)$ one obtains
$\denote{\fT_2(\fT_1(\E))}_{\fL_3}(\nu) \sim
\denote{\fT_1(\E)}_{\fL_2}(\eta)
\sim \denote{\E}_{\fL_1}(\rho)$.
\end{proof}

\begin{definition}{expressiveness}
A language $\fL'$ \emph{can express} or \emph{is at least as
  expressive as} a language $\fL$ \emph{up to $\sim$}, if there exists
a correct translation up to $\sim$ from $\fL$ into $\fL'$.
\end{definition}
\thm{composition} shows that this relation is transitive.
Restricted to languages for which $\sim$ is a congruence, it is even a preorder.

\section{Compositionality}\label{sec:compositionality}

A substitution in $\fL$ is a partial function $\sigma:\V\rightharpoonup\IT_\fL$ from the
variables to the $\fL$-expressions. For a given $\fL$-expression $E\in\IT_\fL$,
$E[\sigma]\in\IT_\fL$ denotes the $\fL$-expression $E$ in which each free occurrence of a
variable $X\in\dom(\sigma)$ is replaced by $\sigma(X)$, while renaming bound variables in $E$
so as to avoid a free variable $Y$ occurring in an expression $\sigma(X)$ ending up being
bound in $E[\sigma]$.
In general, a given expression $E\in\IT_\fL$ can be written in several ways as $F[\sigma]$.
For instance, if $\fL$ features a binary operator $f$, a unary operator $g$ and a constant
$c$, then the term $f(c,g(c))\in\IT_\fL$ can be written as $F[\sigma]$ with
\begin{itemize}
\item $F=f(X,Y)$, $\sigma(X)=c$ and $\sigma(Y)=g(c)$, or
\item $F=f(X,g(Y))$, $\sigma(X)=c$ and $\sigma(Y)=c$, or
\item $F=f(c,g(X))$ and $\sigma(X)=c$.
\end{itemize}
Likewise, in case $\fL$ contains a recursion construct $\textbf{fix}_XS$,
where $S$ is a set of recursion equations $Y=E_Y$, then the expression
$\textbf{fix}_X\{X=f(g(c),g(g(X)))\}$, in which the variable $X$ is bound,
can be written as $F[\sigma]$ with $F=\textbf{fix}_X\{X=f(Y,g(g(X)))\}$ and $\sigma(Y)=g(c)$.

\begin{definition}{prefix}
A term $E\mathbin\in\IT_\fL$ is a \emph{prefix} of a term $F\!$, written $E\mathbin\leq F\!$, if
$F\mathbin{\stackrel{\alpha}=}E[\sigma]$ for some substitution $\sigma$.
Here $\eqa$ denotes \emph{$\alpha$-recursion}, renaming of bound variables
while avoiding capture of free variables.
\end{definition}
Since $E[\textit{id}]=E$, where $\textit{id}:\V\rightarrow\IT_\fL$ is the identity,
and $E[\sigma][\xi]\eqa E[\xi\bullet\sigma]$, where the substitution $\xi\bullet\sigma$ is
given by $(\xi\bullet\sigma)(X)=\sigma(X)[\xi]$, it follows that $\leq$ is reflexive and
transitive, and hence a preorder.
Write $\equiv$ for the kernel of $\leq$, i.e.\ $E\equiv F$ iff $E\leq F \wedge F\leq E$.
If $E\equiv F$ then $E$ can be converted into $F$ by means of an injective renaming of its
variables.

\begin{definition}{head}
An term $H\in\IT_\fL$ is a \emph{head} if $H$ is not a single variable and $E\leq H$
implies that $E$ is single variable or $E\equiv H$.
It is a \emph{head of} another term $F$ if it is a head, as well as a prefix of $F$.
\end{definition}
$f(X,Y)$ is a head of $f(c,g(c))$, and $\textbf{fix}_X\{X=f(Y,g(g(X)))\}$ is a head of
$\textbf{fix}_X\{X=f(g(c),g(g(X)))\}$.

\begin{postulate}{head}
Each expression $E$, if not a variable, has a head, which is unique up to $\equiv$.
\end{postulate}
This is easy to show for each common type of system description language, and I am not
aware of any counterexamples.  However, while striving for maximal generality, I consider
languages with (recursion-like) constructs that are yet to be invented, and in view of
those, this principle has to be postulated rather than derived.
This means that here I consider only languages that satisfy this postulate.
I also limit attention to languages where the meaning of an expression is invariant under
$\alpha$-recursion.
\begin{postulate}{alpha}
If $E\eqa F$ then $\denote{E}_\fL=\denote{F}_\fL$.
\end{postulate}

The semantic mapping $\denote{\ \ }_{\fL}:\IT_{\fL} \rightarrow ((\V\!\rightarrow\bV)\rightarrow\bV)$
extends to substitutions $\sigma$ by $\denote{\sigma}_\fL(\rho)(X):=\denote{\sigma(X)}_\fL(\rho)$
for all $X\mathbin\in\V$ and $\rho:\V\!\!\rightarrow\bV$---here $\sigma$ is extended to a total function by
$\sigma(Y):=Y$ for all $Y\not\in\dom(\sigma)$. Thus $\denote{\sigma}_\fL$ is of type
$(\V\rightarrow\bV)\rightarrow(\V\rightarrow\bV)$, i.e.\ a map from valuations to valuations.
The inductive nature of the semantic mapping $\denote{\ \ }_\fL$ ensures that
\begin{equation}\label{inductive meaning}
\denote{E[\sigma]}_\fL(\rho) = \denote{E}_\fL(\denote{\sigma}_\fL(\rho))
\end{equation}
for all expressions $E\in\IT_\fL$, substitutions $\sigma:\V\rightharpoonup\IT_\fL$ and
valuations $\rho:\V\rightarrow\bV$. In case $E$ is $f(X_1,\ldots,X_n)$ this amounts to
$\denote{f(E_1,\ldots,E_n)}_\fL(\rho) = f^\bD(\denote{E_1}_\fL(\rho),\ldots,\denote{E_n}_\fL(\rho))$,
but (\ref{inductive meaning}) is more general and anticipates language constructs other
than functions, such as recursion.

\begin{definition}{compositionality}
A translation $\fT$ from $\fL$ to $\fL'$ is \emph{compositional} if
$\fT(E[\sigma])\eqa \fT(E)[\fT\circ\sigma]$ for each $E\in\IT_\fL$ and $\sigma:\V\rightharpoonup\IT_\fL$,
and moreover $\fT(X)=X$ for each $X\in\V$.
\end{definition}
In case $E=f(t_1,\ldots,t_n)$ for certain $t_i\in\IT_\fL$ this amounts to
$\fT(f(t_1,\ldots,t_n)) \eqa  E_f(\fT(t_1),\ldots,\fT(t_n))$,
where $E_f:=\fT(f(X_1,\ldots,X_n))$ and
$E_f(u_1,\ldots,u_n)$ denotes the result of the simultaneous
substitution in this expression of the terms $u_i\in\IT_{\fL'}$ for the free variables $X_i$,
for $i=1,\ldots,n$. Again, \df{compositionality} is more general and anticipates language
constructs other than functions, such as recursion.

\begin{theorem}{compositionality}
If any correct translation from $\fL$ to $\fL'$ up to $\sim$ exists,
then there exists a compositional translation that is correct up to $\sim$.
\end{theorem}

\begin{proof}
Pick a representative from each $\equiv$-equivalence class of terms.
With \emph{the head of an expression} $E$ I mean the chosen representative out of the
$\equiv$-equivalence class of heads of $E$. Now each term $E\notin\V$ can uniquely be written as
$H[\sigma]$, with $H$ the head of $E$ and $\dom(\sigma)$ the set of free variables of $H$.

Given a correct translation $\fT_0$, define the translation $\fT$ inductively by
\begin{center}
\begin{tabular}{ll}
$\fT(X) := X$ & for $X\in \V$\\
$\fT(E) := \fT_0(H)[\fT\circ\sigma]$ & when $E\eqa H[\sigma]$ as stipulated above.
\end{tabular}
\end{center}
First I show that $\fT$ is compositional, using induction on $E$.
So let $E\in\IT_\fL$ and $\xi:\V\rightarrow\IT_\fL$.
I have to show that $\fT(E[\xi])\eqa \fT(E)[\fT\circ\xi]$.
The case $E\in\V$ is trivial, so let $E\eqa H[\sigma]$.
For each free variable $X$ of $H$, $\sigma(X)$ is a proper subterm of $E$,
so by the induction hypothesis $\fT(\sigma(X)[\xi])\eqa \fT(\sigma(X))[\fT\circ\xi]$.
Thus
$\begin{array}[t]{@{}l@{~}ll@{}}
(\fT\!\circ(\xi\bullet\sigma))(X)
&= \fT((\xi\bullet\sigma)(X))
& \mbox{by definition of functional composition $\circ$}
\\
&= \fT(\sigma(X)[\xi])
& \mbox{by definition of the relation $\bullet$ between substitutions}
\\
&\eqa \fT(\sigma(X))[\fT\circ\xi]
& \mbox{by induction, derived above; trivial if $X\not\in\dom(\sigma)$}
\\
&= ((\fT\!\circ\xi)\bullet(\fT\!\circ\sigma))(X)
& \mbox{by definition of the relations $\circ$ and $\bullet$.}
\end{array}$\\
This shows that the substitutions $\fT\!\circ(\xi\bullet\sigma)$ and
$(\fT\!\circ\xi)\bullet(\fT\!\circ\sigma)$ are equal up to $\alpha$-recursion, from which it
follows that that $F[\fT\circ(\xi\bullet\sigma)] \eqa (F[\fT\circ\sigma])[\fT\circ\xi]$
for all terms $F\in\IT_{\fL'}$.
\\
$\begin{array}[b]{@{}l@{~}ll@{}}
\mbox{Hence}~ \fT(E[\xi])
&\eqa \fT(H[\sigma][\xi])
& \mbox{since $E\eqa H[\sigma]$.}
\\
&\eqa \fT(H[\xi\bullet\sigma])
& \mbox{by the identity used already in proving transitivity of $\leq$}
\\
&= \fT_0(H)[\fT\circ(\xi\bullet\sigma)]
& \mbox{by definition of $\fT$}
\\
&\eqa (\fT_0(H)[\fT\circ\sigma])[\fT\circ\xi]
& \mbox{derived above}
\\
&= \fT(H[\sigma])[\fT\circ\xi]
& \mbox{by definition of $\fT$}
\\
&\eqa \fT(E)[\fT\circ\xi]
& \mbox{since $E\eqa H[\sigma]$.}
\end{array}$

It remains to be shown that $\fT$ is correct up to $\sim$, i.e.\
that $\denote{\fT(\E)}_{\fL'}(\eta) \sim \denote{\E}_\fL(\rho)$ for all
terms $\E\in \IT_\fL$ and all valuations $\eta:\V\rightarrow\bV'$ and
$\rho:\V\rightarrow\bV$ with $\eta\sim\rho$.
Let $\eta$ and $\rho$ be such valuations.
I proceed with structural induction on $E$.
When handling a term $E\eqa H[\sigma]$, $\sigma(X)$ is a proper subterm of $E$ for each free variable $X$ of $H$.
So by the induction hypothesis $\denote{\fT(\sigma(X)}_{\fL'}(\eta) \sim \denote{\sigma(X)}_\fL(\rho)$.
The valuation $\denote{\sigma}_\fL(\rho)$ is defined such that
$\denote{\sigma}_\fL(\rho)(X)=\denote{\sigma(X)}_\fL(\rho)$ for each $X\in\V$.
Likewise, $\denote{\fT\circ\sigma}_{\fL'}(\eta)(X)=\denote{\fT(\sigma(X)}_{\fL'}(\eta)$ for each $X\in\V$.
Hence $\denote{\fT\circ\sigma}_{\fL'}(\eta) \sim \denote{\sigma}_\fL(\rho)$.\hfill(*)
\begin{list}{$\bullet$}{\leftmargin 10pt}
\item $\denote{\fT(X)}_{\fL'}(\eta) = \denote{X}_{\fL'}(\eta)
\begin{array}[t]{@{~}ll}
=\eta(X) & \mbox{by definitions of $\fT$ and $\denote{\ \ }_{\fL'}$} \\
\sim\rho(X) & \mbox{since $\eta\sim\rho$} \\
=\denote{X}_{\fL}(\rho) & \mbox{by definition of $\denote{\ \ }_{\fL}$.}
\end{array}$
\item $\denote{\fT(H[\sigma])}_{\fL'}(\eta)
\begin{array}[t]{@{~}ll@{}}
= \denote{\fT_0(H)[\fT\circ\sigma]}_{\fL'}(\eta) & \mbox{by definition of $\fT$} \\
= \denote{\fT_0(H)}_{\fL'}(\denote{\fT\circ\sigma}_{\fL'}(\eta)) & \mbox{by (\ref{inductive meaning})} \\
\sim \denote{H}_\fL(\denote{\sigma}_\fL(\rho)) & \mbox{by (*) above, as $\fT_0$ is a correct translation\,~~~~~~~~~~~} \\
= \denote{H[\sigma]}_\fL(\rho)  & \mbox{by (\ref{inductive meaning}).}
\hfill\mbox{\qed}
\end{array}$
\end{list}
\end{proof}

\noindent
Hence, for the purpose of comparing the expressive power of languages,
correct translations between them can be assumed to be compositional.

\section{Comparing the expressive power of CCS and CSP}

As an application of my approach, in this section I quantify the degree to which the
parallel composition of CSP can be expressed in CCS\@. It turns out that there exists a
correct translation up to trace equivalence, but not up to the version of weak bisimilarity
equivalence that takes divergence into account. This combination of an encoding and a
separation result is typical when comparing system description languages. Here we see that
for applications where divergence and branching time are a concern, the CSP parallel
composition cannot be encoded in CCS; however, when linear time reasoning is all that
matters, it can.

\subsection{CCS}

CCS \cite{Mi90ccs} is parametrised with a set ${\cal A}$ of {\em names}.
The set $\bar{\cal A}$ of {\em co-names} is $\bar\mathcal{A}:=\{\bar{a}
\mid a\in {\cal A}\}$, and $\mathcal{L}:=\mathcal{A} \cup \bar\mathcal{A}$ is the set of
\emph{labels}.  The function $\bar{\cdot}$ is extended to $\mathcal{L}$ by declaring
$\bar{\bar{\mbox{$a$}}}=a$. Finally, \plat{$Act := \mathcal{L}\dcup \{\tau\}$} is the set of
{\em actions}. Below, $a$, $b$, $c$, \ldots range over $\mathcal{L}$ and
$\alpha$, $\beta$ over $Act$.
A \emph{relabelling function} is a function $f:\mathcal{L}\rightarrow \mathcal{L}$ satisfying
$f(\bar{a})=\overline{f(a)}$; it extends to $Act$ by $f(\tau):=\tau$.
Let $\mathcal{X}$ be a set $X$, $Y$, \ldots of \emph{process variables}.
The set $\mathcal{E}$ of CCS terms or \emph{process expressions} is the smallest set
including:
\begin{center}
\begin{tabular}{lll}
$\alpha.E$  & for $\alpha\in Act$ and $E\in\mathcal{E}$ & \emph{prefixing}\\
$\sum_{i\in I}E_i$ & for $I$ an index set and $E_i\in\mathcal{E}$ & \emph{choice} \\
$E|F$ & for $E,F\in\mathcal{E}$ & \emph{parallel composition} \\
$E\backslash L $ & for $L\subseteq\mathcal{L}$ and $E\in\mathcal{E}$ & \emph{restriction} \\
$E[f]$ & for $f$ a relabelling function and $E\in\mathcal{E}$ & \emph{relabelling} \\
$X$ & for $X\in\mathcal{X}$ & a \emph{process variable} \\
$\textbf{fix}_XS$ & for $S:\mathcal{X}\rightharpoonup \mathcal{E}$
and $X\in \dom(S)$ & \emph{recursion}.
\end{tabular}
\end{center}
One writes $E_1+E_2$ for $\sum_{i\in I}E_i$ with $I=\{1,2\}$, and $0$ for $\sum_{i\in \emptyset}E_i$.
A partial function $S:\mathcal{X}\rightharpoonup \mathcal{E}$ is called a
\emph{recursive specification}.
The variables in its domain $\dom(S)$ are called \emph{recursion variables} and
the equations $Y=S(Y)$ for $Y\in\dom(S)$ \emph{recursion equations}.
A recursive specification $S:\mathcal{X}\rightharpoonup \mathcal{E}$ is traditionally
written as $\{Y=S(Y)\mid Y\in \dom(S)\}$.

\begin{table}[t]
\begin{center}
\framebox{$\begin{array}{ccc}
\alpha.E \goto{\alpha} E &
\displaystyle\frac{E_j \goto{\alpha} E_j'}{\sum_{i\in I}E_i \goto{\alpha} E_j'}\makebox[0pt][l]{~~($j\in I$)}
 \\[4ex]
\displaystyle\frac{E\goto{\alpha} E'}{E|F \goto{\alpha} E'|F} &
\displaystyle\frac{E\goto{a} E' ,~ F \goto{\bar{a}} F'}{E|F \goto{\tau} E'| F'} &
\displaystyle\frac{F \goto{\alpha} F'}{E|F \goto{\alpha} E|F'}\\[4ex]
\displaystyle\frac{E \goto{\alpha} E', ~\alpha\not\in L\cup\bar{L}}{E\backslash L \goto{\alpha}
E'\backslash L} &
\displaystyle\frac{E \goto{\alpha} E'}{E[f] \goto{f(\alpha)} E'[f]} &
\displaystyle\frac{S(X)[\textbf{fix}_YS/Y]_{Y\in \dom(S)} \goto{\alpha} E}{\textbf{fix}_XS\goto{\alpha}E}
\end{array}$}
\end{center}
\caption{Structural operational semantics of CCS}
\label{tab:CCS}
\end{table}

CCS is traditionally interpreted in the domain $\T_{\rm CCS}$ of closed CCS expressions up
to $\alpha$-recursion.
Hence a valuation $\rho:\V\rightarrow\T_{\rm CCS}$, valuating each variable as a closed
CCS expression, is just a closed substitution. The semantic mapping
$\denote{\ \ }_{\rm CCS}$ is given by $\denote{E}_{\rm CCS}(\rho) := E[\rho]$---a CCS
expression $E$ evaluates, under the valuation $\rho:\V\rightarrow\T_{\rm CCS}$, to the
result of performing the substitution $\rho$ on $E$. In fact, this is a common way to
provide many system description languages with a semantics. Consequently, the distinction
between syntax and semantics can, to a large extent, be dropped. It is for this reason
that the semantic interpretation function $\denote{\ \ }$ rarely occurs in papers on
CCS-like languages.

The ``real'' semantics of CCS is given by the labelled transition relation
$\mathord\rightarrow \subseteq \T_{\rm CCS}\times Act \times\T_{\rm CCS}$ between
closed CCS expressions. The transitions \plat{$p\goto{\alpha}q$} with $p,q\in\T_{\rm CCS}$
and $\alpha\in Act$ are derived from the rules of \tab{CCS}.
Formally a transition \plat{$p\goto{\alpha}q$} is part of the transition relation of CCS
if there exists a well-founded, upwards branching tree (a \emph{proof} of the transition)
of which the nodes are labelled by transitions, such that
\begin{itemize}\itemsep 0pt
\item the root is labelled by $p\goto{\alpha}q$, and
\item if $\varphi$ is the label of a node $n$ and $K$ is the set of labels of the nodes
  directly above $n$, then $\frac{K}{\varphi}$ is a rule from \tab{CCS}, with closed CCS
  expressions substituted for the variables $E,F,\ldots$.
\end{itemize}

\subsection{CSP}

CSP \cite{BHR84,OH86,BR85,Ho85} is parametrised with a set ${\cal A}$ of {\em communications\/};
\plat{$Act := \mathcal{A}\dcup \{\tau\}$} is the set of {\em actions}. Below, $a$, $b$
range over $\mathcal{A}$ and $\alpha$, $\beta$ over $Act$.
The set $\mathcal{E}$ of CSP terms is the smallest set including:
\begin{center}
\begin{tabular}{lll}
$\stp$ && \emph{inaction} \\
$\dv$ && \emph{divergence} \\
$(a\rightarrow E)$  & for $a\in \mathcal{A}$ and $E\in\mathcal{E}$ & \emph{prefixing}\\
$E \obox F$ & for $E,F\in\mathcal{E}$ & \emph{external choice} \\
$E \sqcap F$ & for $E,F\in\mathcal{E}$ & \emph{internal choice} \\
$E\|_AF$ & for $E,F\in\mathcal{E}$ and $A\subseteq{\cal A}$ & \emph{parallel composition} \\
$E\conceal b$ & for $b\in\mathcal{A}$ and $E\in\mathcal{E}$ & \emph{concealment} \\
$f(E)$ & for $E\in\mathcal{E}$ and $f:Act\rightarrow Act$ with $f(\tau)=\tau$ and
  $f^{-1}(a)$ finite& \emph{renaming} \\
$X$ & for $X\in\mathcal{X}$ & a \emph{process variable} \\
$\mu X\cdot E$ & for $E\in \mathcal{E}$ and $X\in \V$ & \emph{recursion}.
\end{tabular}
\end{center}
As in \cite{OH86}, I here leave out the guarded choice $(x:B\rightarrow P(x))$ and the
constant \mbox{\sc run} of \cite{BHR84}, and the inverse image and sequential composition operator,
with constant \mbox{\sc skip}, of \cite{BHR84,BR85}.
The semantics of CSP was originally given in quite a different way \cite{BHR84,BR85}, but \cite{OH86} provided an
operational semantics of CSP in the same style as the one of CCS, and showed its
consistency with the original semantics. It is this operational semantics I will use here;
it is given by the rules in \tab{CSP}. Let $\mathcal{L}:=\mathcal{A}$.
\begin{table}[htb]
\begin{center}
\framebox{$\begin{array}{ccccc}
\dv\goto{\tau}\dv &
(a\rightarrow E) \goto{a} E &
E \sqcap F \goto{\tau} E &
E \sqcap F \goto{\tau} F \\[2ex]
\displaystyle\frac{E\goto{a} E'}{E\obox F \goto{a} E'} &
\displaystyle\frac{F\goto{a} F'}{E\obox F \goto{a} F'} &
\displaystyle\frac{E\goto{\tau} E'}{E\obox F \goto{\tau} E'\obox F} &
\displaystyle\frac{F\goto{\tau} F'}{E\obox F \goto{\tau} E\obox F'} \\[4ex]
\displaystyle\frac{E\goto{\alpha} E'~~{\scriptstyle(\alpha\notin A)}}{E\|_AF \goto{\alpha} E'\|_AF} &
\multicolumn{2}{c}{
\displaystyle\frac{E\goto{a} E'~~F\goto{a} F'~~{\scriptstyle(a\in A)}}{E\|_AF \goto{a} E'\|_AF'}} &
\displaystyle\frac{F\goto{\alpha} F'~~{\scriptstyle(\alpha\notin A)}}{E\|_AF \goto{\alpha} E\|_AF'} \\[4ex]
\displaystyle\frac{E \goto{b} E'}{E\conceal b \goto{\tau} E'\conceal b} &
\displaystyle\frac{E \goto{\alpha} E'~~{\scriptstyle(\alpha\neq b)}}{E\conceal b \goto{\alpha} E'\conceal b} &
\displaystyle\frac{E \goto{\alpha} E'}{f(E) \goto{f(\alpha)} f(E')} &
\multicolumn{2}{c}{\mu X\cdot E \goto{\tau} E[\mu X\cdot E/X]}
\end{array}$}
\end{center}
\caption{Structural operational semantics of CSP}
\label{tab:CSP}
\end{table}

\subsection{Trace semantics and convergent weak bisimilarity}

I will compare the expressive power of CCS and CSP up two semantic equivalences: a linear
time and a branching time equivalence. For the former I take \emph{trace equivalence}
\cite{Ho80} and for the latter a version of weak bisimilarity that takes
divergence into account \cite{HP80,Sti87,Ab87,Wa90}---called \emph{convergent weak bisimilarity}
in \cite{vG93}. Unlike the standard weak bisimilarity of \cite{Mi90ccs}, this relation is
finer than the failures-divergences semantics of \cite{BHR84,OH86,BR85,Ho85}.

The relation $\mathord\Rightarrow \subseteq \T_{\rm CCS}\times \fL^* \times\T_{\rm CCS}$
is the transitive closure of $\rightarrow$ that abstracts from $\tau$-steps. Formally,
$\dto{}$ is the transitive closure of \plat{$\goto{\tau}$}
and $p\dto{a_1\cdots a_n}q$ for $n\mathbin\geq 0$ holds iff there are $p_0,p_1,\ldots,p_n$
with \mbox{$p_0\mathbin=p$}, $p_{i-1} \dto{}\goto{a_i}p_i$ for $i=1,\ldots,n$, and $p_n\dto{} q$.
Below, $\T$ is a set that contains $\T_{\rm CCS}$ and $\T_{\rm CSP}$.

\begin{definition}{traces}
The set $T(p)\subseteq \mathcal{L}^*$ of \emph{traces} of a process $p\mathbin\in \T$ is given by
$\s\in T(p)$ iff $\exists p'.~p\dto{\s} p'$.
Two processes $p,q\in\T$ are \emph{trace equivalent} if $T(p)=T(q)$.
\end{definition}

\begin{definition}{weak bisimilarity}
A relation $\B \subseteq \T\times\T$ is a \emph{weak bisimulation} \cite{Mi90ccs} if
\begin{itemize}
\item for any $p,p',q\in\T$ and $\s\in \mathcal{L}^*$ with $p\B q$ and
  $p\dto{\s}p'$, there is a $q'$ with $q\dto{\s}q'$ and $p'\B q'$,
\item for any $p,q,q'\in\T$ and $\s\in \mathcal{L}^*$ with $p\B q$ and
  $q\dto{\s}q'$, there is a $p'$ with $p\dto{\s}p'$ and $p'\B q'$.
\end{itemize}
Two processes $p,q\in\T$ are \emph{weakly bisimilar}, $p\bis{w} q$, if they are related by a
weak bisimulation.
\end{definition}
All we need to know about the \emph{convergent} weak bisimilarity ($\bis[\downarrow]{w}$)
is that a process that has a divergence cannot be related to a divergence-free process,
and that restricted to divergence-free processes it coincides with weak bisimilarity. Here
a process \emph{has a divergence} if it can do an infinite sequence of transitions that
from some point onwards are all labelled $\tau$.

Trace equivalence and (convergent) weak bisimilarity are congruences for CSP\@.  The
(convergent) weak bisimilarity fails to be a congruence for the $+$ of CCS, a problem that
is commonly solved by taking its congruence closure. I do not need to do this when
translating CSP into CCS, because correct translations need not be a congruence for the
whole target language.

Note that even when restricting CCS to just $0$, action prefixing and $+$, there is no
correct translation of this language into CSP up to the congruence closure of
$\bis[\downarrow]{w}$---this is a direct consequence of \cor{congruence}.

\subsection{A correct translation of CSP into CCS up to trace equivalence}\label{sec:CSP-CCS}

For any choice of a CSP set of communications $\mathcal{A}$, I create a CCS set of names
$\mathcal{B}$ and construct a translation from CSP with communications from $\mathcal{A}$
into CCS with names from $\mathcal{B}$.

Let $\mathcal{B}:= \{a,a',a''\mid a\in \mathcal{A}\}$, consisting of 3 disjoint copies of $\mathcal{A}$.
For $A\subseteq\mathcal{A}$, let $S_A$ be the recursive specification given by the single CCS equation
$\displaystyle \{X\mathbin=\!\sum_{a\in  A}\bar{a}.a'.a''.a'.X +
                             \!\!\!\!\sum_{a\in \mathcal{A}\!-A}\!\!\!\bar{a}.a''.X\}$
and $S'_A$ be the recursive specification given by the single CCS equation
$\displaystyle \{X\mathbin=\!\sum_{a\in  A}\bar{a}.\bar{a}'.\bar{a}'.X +
                             \!\!\!\!\sum_{a\in \mathcal{A}\!-A}\!\!\!\bar{a}.a''.X\}$.
Now, up to trace equivalence, and assuming that $P$ features names from $\mathcal{A}$ only,
$(P|\textbf{fix}_XS_A)\backslash \mathcal{A}$ is a process
that differs from $P$ by the replacement of each $a$-transition by a sequence of
transitions $a' a'' a'$ if $a\in A$, and by the single transition $a''$ otherwise.
Likewise,  $(P|\textbf{fix}_XS'_A)\backslash \mathcal{A}$ differs from $P$ by the
replacement of each $a$-transition by a $\bar{a}'\bar{a}'$ if $a\in A$, and $a''$ otherwise.
Let $\mathcal{A}':=\{a' \mid a\in\mathcal{A}\}$, and let the relabelling function $f$ be
such that $f(a'')=a$. Then the following is a correct translation of CSP into CCS up to
trace equivalence.

$\fT(X)=X$

$\fT(\mu X\cdot E) = \textbf{fix}_X\{X=\fT(E)\}$

$\fT(a\rightarrow E)=a.\fT(E)$

$\fT(\stp)=\fT(\dv)=0$

$\fT(E\sqcap F) = \fT(E\obox F) = \fT(E) + \fT(F)$

$\fT(E\conceal b) = (\fT(E) | \textbf{fix}_X\{X=\bar{b}.X\})\backslash \{b\}$

$\fT(f(E)) = \fT(E)[f]$

$\displaystyle \fT(E\|_AF) = \left(
\left( (\fT(E)|\textbf{fix}_XS_A)\backslash \mathcal{A}
\big| (\fT(F)|\textbf{fix}_XS'_A)\backslash \mathcal{A}
\right)\backslash \mathcal{A}'\right)[f]$

\subsection{The untranslatability of CSP into CCS up to convergent weak bisimilarity}
\label{sec:no translation}

In this section I show that there is no translation of CSP into CCS up to convergent weak
bisimilarity. Suppose that $\fT$ is such a translation.
Let $\rho:\V\rightarrow \T_{\rm CSP}$ and $\eta:\V\rightarrow \T_{\rm CCS}$ satisfy
$\rho(X)=\rho(Y)=(b\rightarrow\stp)\obox (b\rightarrow(c\rightarrow\stp))$ and
$\eta(X)=\eta(Y)=b.0+b.c.0$. Then \plat{$\rho\bis[\downarrow]{w}\eta$}.
So\vspace{-3pt}
\[
\fT(X\|_{\{b,c\}}Y)[\eta] = \denote{\fT(X\|_{\{b,c\}}Y)}_{\rm CCS}(\eta)
\bis[\downarrow]{w} \denote{X\|_{\{b,c\}}Y}_{\rm CSP}(\rho) \bis[\downarrow]{w} b.0+b.c.0.
\]
Let $\nu:\V\rightarrow \T_{\rm CCS}$ satisfy $\nu(X)=\nu(Y)=b.0$. By the same reasoning as above\vspace{-3pt}
\[
\fT(X\|_{\{b,c\}}Y)[\nu] \bis[\downarrow]{w} b.0.
\]
Since $b.0$ has no divergence, neither does $\fT(X\|_{\{b,c\}}Y)[\nu]$, so
there must be a state $p\in\T_{\rm CCS}$ with
\plat{$\fT(X\|_{\{b,c\}}Y)[\nu] \dto{} p \gonotto{\tau}$}.
By \cite[Proposition 7.1 (or 8)]{BFG04}, it follows from the operational semantics of CCS that
if $E[\sigma] \goto{\alpha} q$ for $E\mathbin\in\IT_{\rm CCS}$, $\sigma:\V\rightarrow \T_{\rm  CCS}$
and $q\in\T_{\rm CCS}$, then $q$ must have the form $F[\sigma']$ with $F\in\IT_{\rm CCS}$
and for each variable $W$ that occurs free in $F$ there is a variable $Z$ that occurs free
in $E$, such that either $\sigma(Z)=\sigma'(W)$ or
\plat{$\sigma(Z)\goto{\beta}\sigma'(W)$} for some $\beta\in Act$\footnote{In general
  multiple occurrences of $Z$ in $E$ may give rise to different associated variables $W$
  in $F$.}---moreover, $F$ depends on $E$ and on the existence of the $\beta$-transitions,
but not any other property of $\sigma$. So, for some $n\geq 0$,
$$\fT(X\|_{\{b,c\}}Y)[\nu] \goto{\tau} E_1[\nu_1]  \goto{\tau} E_2[\nu_2] \goto{\tau}\ldots
\goto{\tau} E_n[\nu_n] \gonotto{\tau}$$
where, for any free variable $Z$ of $E_i$, $\nu_i(Z)$ is either $0$ or $b.0$.
This execution path can be simulated by
$$\fT(X\|_{\{b,c\}}Y)[\eta] \goto{\tau} E_1[\eta_1]  \goto{\tau} E_2[\eta_2] \goto{\tau}\ldots
\goto{\tau} E_n[\eta_n] \gonotto{\tau}$$
where $\eta_i(Z)=b.0+b.c.0$ iff $\nu_i(Z)=b.0$ and $\eta_i(Z)=0$ iff
$\nu_i(Z)=0$---i.e.\ always choosing $\eta(Z)\goto{b}0$ over \plat{$\eta(Z)\goto{b}c.0$}.
By the properties of $\bis[\downarrow]{w}$, \plat{$E_n[\eta_n] \bis[\downarrow]{w} b.0+b.c.0$}.
So there is a process $E_{n+1}[\eta_{n+1}]$ with
\plat{$E_n[\eta_n]\goto{b}E_{n+1}[\eta_{n+1}]\dto{}\goto{c}$}.
It must be that $E_{n+1}[\eta_{n+1}]\bis[\downarrow]{w}c.0$.

The only rule in the structural operational semantics of CCS that has multiple premises
has a conclusion with label $\tau$. Furthermore, any rule with a $\tau$-labelled premise,
has a $\tau$-labelled conclusion. Hence, since the transition
\plat{$E_n[\eta_n]\goto{b}E_{n+1}[\eta_{n+1}]$} is not labelled $\tau$, its proof has only
one branch. This branch could stem from a transition from $\eta(X)$ or from $\eta(Y)$, but
not both. W.l.o.g.\ I assume it does not stem from $\eta(X)$.

Let $\xi:\V\rightarrow \T_{\rm CCS}$ satisfy $\xi(X)=b.0$ and $\xi(Y)=b.0+b.c.0$.
Since in the proofs of the transitions in the above path from $\fT(X\|_{\{b,c\}}Y)[\eta]$ the transition
\plat{$\eta(X)\goto{b}c.0$} is never used, that path can be simulated by\vspace{-2ex}
$$\fT(X\|_{\{b,c\}}Y)[\xi] \goto{\tau} E_1[\xi_1]  \goto{\tau} E_2[\xi_2] \goto{\tau}\ldots
\goto{\tau} E_n[\xi_n] \goto{b} E_{n+1}[\xi_{n+1}].$$
Note that \plat{$\fT(X\|_{\{b,c\}}Y)[\xi] \bis[\downarrow]{w} b.0$}.
Due to the properties of $\bis[\downarrow]{w}$ the above derivation can be extended with\vspace{-1ex}
$$E_{n+1}[\xi_{n+1}]  \goto{\tau} E_{n+2}[\xi_{n+2}] \goto{\tau}\ldots \goto{\tau} E_{n+k}[\xi_{n+k}]$$
ending in a \emph{deadlock} state, where no further transitions are possible.
This derivation, in turn, can be simulated by\vspace{-1ex}
$$E_{n+1}[\eta_{n+1}]  \goto{\tau} E_{n+2}[\eta_{n+2}] \goto{\tau}\ldots \goto{\tau} E_{n+k}[\eta_{n+k}],$$
still ending in a deadlock state. This contradicts $E_{n+1}[\eta_{n+1}]\bis[\downarrow]{w}c.0$.
\hfill $\Box$

\section{Valid translations up to a preorder}\label{sec:respects}

Let $\fL$ and $\fL'$ be languages with
$\denote{\ \ }_{\fL}:\IT_{\fL} \rightarrow ((\V\rightarrow\bV)\rightarrow\bV)$
and
$\denote{\ \ }_{\fL'}:\IT_{\fL'} \rightarrow ((\V\rightarrow\bV')\rightarrow\bV')$.
In this section I explore an alternative for the notion of a correct translation up to an
equivalence $\sim$. This alternative doesn't have a build-in requirement that $\sim$ must
be a congruence for $\fL$;\footnote{Moreover, it may be a preorder rather than an equivalence.}
however it only deals with semantic values denotable by closed terms.

Let $\T_\fL$ be the set of closed $\fL$-expressions, i.e.\ having no
free variables. The meaning $\denote{\p}_\fL(\rho)$ of a closed term
$\p\in \T_\fL$ is independent of the valuation $\rho:\V\rightarrow\bV$, and hence
denoted $\denote{\p}_\fL$.

\begin{definition}{respects}
A translation $\fT$ from $\fL$ into $\fL'$ \emph{respects} $\sim$ if (\ref{related}) holds
and $\denote{\fT(P)}_{\fL'}(\eta) \sim \denote{P}_\fL$ for all closed $\fL$-expressions $P\in\T_\fL$
and all valuations $\eta:\V\rightarrow\bU$, with $\bU:=\{v\in\bV'\mid \exists v\in \bV.~v'\sim v\}$.
\end{definition}

\begin{observation}{correct respects}
If $\fT$ is a correct translation from $\fL$ into $\fL'$ up to $\sim$,
then it respects $\sim$.
\end{observation}
Usually one employs translations $\fT$ with the property that for any
$E\in\IT_\fL$ any free variable of $\fT(E)$ is also a free variable of $E$---I call these
\emph{free-variable respecting translations}, or \emph{fvr-translations}.
If there is at least one $Q\in\T_{\fL'}$ with $\denote{Q}_{\fL'}\in\bU$, then any
translation $\fT$ from $\fL$ into $\fL'$ can be modified to an fvr-translation $\fT^\circ$
from $\fL$ into $\fL'$, namely by substituting $Q$ for all free variables of
$\fT(E)$ that are not free in $E$.  This modification preserves the properties of
respecting $\sim$ and of being correct up to $\sim$.
An fvr-translation $\fT$ from $\fL$ into $\fL'$ \emph{respects} $\sim$ iff
$\denote{\fT(P)}_{\fL'} \sim \denote{P}_\fL$ for all closed $\fL$-expressions $P\in\T_\fL$.

\begin{observation}{fvr}
Let $\fT: \IT_\fL \rightarrow \IT_{\fL'}$ be an fvr-translation from $\fL$ into $\fL'$, and let
$\sim,\approx$ be equivalences (or preorders) on a class $\bZ \subseteq \bV\cup\bV'$, with
$\sim$ finer than $\approx$. If $\fT$ respects $\sim$, then it also respects $\approx$.

The identity is a $\sim$-respecting fvr-translation from any language into itself.

If $\sim$-respecting fvr-translations exists from $\fL_1$ into $\fL_2$ and from $\fL_2$ into $\fL_3$,
then there is a $\sim$-respecting fvr-translation from $\fL_1$ into $\fL_3$.
\end{observation}

\noindent
Respecting an equivalence or preorder is a very weak correctness requirement for translations.
In spite of the separation result of \sect{no translation}, there trivially exists a
translation from CSP to CCS that respects $\bis[\downarrow]{w}$, or even strong bisimilarity.
This follows from the observation that---thanks to the arbitrary index sets
$I$ and $\dom(S)$ that may be used for choice and recursion---up to $\bis[\downarrow]{w}$ every
process graph is denotable by a CCS expression.
In particular, compositionality is in no way implied by respect for an equivalence.
It therefore makes sense to add compositionality as a separate requirement.
The following shows that also the notion of a compositional $\sim$-respecting
transition is a bit too weak.

\begin{example}{undenotable}
Let $\fL'$ be the language CCS without the recursion construct, but interpreted in a
domain of arbitrary process graphs (similar to the graph model of ACP \cite{BW90}).
Let $\fL$ be the same language, but with an extra operator $\_\!\_/\mathcal{L}$
that relabels all transitions into $\tau$. The compositional translation $\fT$ from $\fL$
into $\fL'$ with $\fT(X/\mathcal{L}):=0$ respects \plat{$\bis[\downarrow]{w}$}.
This is because the interpretation of any closed $\fL$-expression is a process graph
without infinite paths, and after relabelling all transitions into $\tau$ such a graph is
equivalent to $0$. Yet, there are process graphs $G$---those with infinite paths---that
cannot be denoted by closed $\fL$-expressions, and for which $G/\mathcal{L} \not\!\bis[\downarrow]{w} 0$,
demonstrating that $\fT$ should not be seen as a valid translation.
\end{example}
Based on this, I add the denotability of all semantic values as a requirement
of a valid translation.

\begin{definition}{valid}
A translation $\fT$ from $\fL$ into $\fL'$ is \emph{valid up to} $\sim$ if
it is compositional and respects $\sim$, while $\fL$ satisfies\vspace{-1ex}
\begin{equation}\label{denotable}
\forall v\in\bV.~\exists \p\in\T_\fL.~\denote{\p}_\fL=v\;.
\end{equation}
\end{definition}
The following theorem (in combination with \thm{compositionality} and \obs{correct respects})
shows that this notion of a valid translation is consistent with the notion of a correct
translation, and can be seen as extending that notion to situations where $\sim$ is not
known to be a congruence.

\begin{theorem}{valid correct}
Let $\fT: \IT_\fL \rightarrow \IT_{\fL'}$ be a translation from $\fL$ into $\fL'$,
and $\sim$ be a congruence for $\fT(\fL)$.
If $\fT$ is valid up to $\sim$, then it is correct up to $\sim$.
\end{theorem}

\begin{proof}
Suppose $\fT$ is valid up $\sim$.
Then $\denote{\fT(P)}_{\fL'}(\eta) \sim \denote{P}_\fL$ for all
all closed $\fL$-expressions $P\in \T_\fL$ and all valuations $\eta:\V\rightarrow\bU$.
To establish that $\fT$ is correct up to $\sim$,
let $\E\in \IT_\fL$ and let $\eta:\V\rightarrow\bV'$ and $\rho:\V\rightarrow \bV$ be
valuations with $\eta\sim\rho$. So $\eta:\V\rightarrow\bU$.
I need to show that $\denote{\fT(\E)}_{\fL'}(\eta) \sim \denote{\E}_\fL(\rho)$.

Let $\sigma:\V\rightarrow\T_\fL$ be a substitution with $\denote{\sigma(X)}_\fL=\rho(X)$
for all $X\in\V$---such a substitution exists by (\ref{denotable}).
Furthermore, define $\nu:\V\rightarrow\bV'$ by
$\nu(X):=\denote{\fT(\sigma(X))}_{\fL'}(\eta)$ for all $X\in\V$.
Since $\fT$ respects $\sim$ I have $\nu(X)\sim\rho(X)$ for all $X\in\V$; thus
$\eta\sim\rho\sim\nu$ and also $\nu:\V\rightarrow\bU$.\\
$\begin{array}[b]{@{}l@{~}ll@{}}
\mbox{Hence}~\denote{\fT(\E)}_{\fL'}(\eta)
& \sim \denote{\fT(\E)}_{\fL'}(\nu) & \mbox{since $\sim$ is a congruence for $\fT(\fL)$}\\
& = \denote{\fT(E)}_{\fL'}(\denote{\fT\circ\sigma}_{\fL'}(\eta)) & \mbox{expanding the definition of $\nu$} \\
& = \denote{\fT(E)[\fT\circ\sigma]}_{\fL'}(\eta) & \mbox{by (\ref{inductive meaning})}\\
& = \denote{\fT(E[\sigma])}_{\fL'}(\eta) & \mbox{by compositionality of $\fT$} \\
& \sim \denote{E[\sigma]}_{\fL} & \mbox{since $\fT$ respects $\sim$}\\
& = \denote{\E}_\fL(\denote{\sigma}_\fL) & \mbox{by (\ref{inductive meaning})} \\
& = \denote{\E}_\fL(\rho) & \mbox{by definition of $\rho$.}
\end{array}$
\end{proof}

\section{Related work}

The greatest expressibility result presented so far is by De Simone \cite{dS85}, who
showed that a wide class of languages, including CCS, SCCS, CSP and ACP, are expressible
up to strong bisimulation equivalence in {\sc Meije}.  Vaandrager \cite{Va93} established
that this result crucially depends on the use of unguarded recursion, and its noncomputable
consequences. {\em Effective} versions of CCS, SCCS, {\sc Meije} and ACP, not using
unguarded recursion, are incapable of expressing all effective De Simone
languages. Nevertheless, \cite{vG94a} isolated a \emph{primitive effective} dialect of ACP
(featuring primitive recursive renaming operators) in which a large class of primitive
effective languages, including primitive effective versions of CCS, SCCS, CSP and {\sc Meije},
can be encoded. All these results fall within the scope of the notion of translation and
expressibility from \cite{Bo85} and \cite{vG94a}, and use strong bisimulation as underlying
equivalence. 

In the last few years, a great number of encodability and separation results have
appeared, comparing CCS, Mobile Ambients, and several versions of the $\pi$-calculus (with
and without recursion; with mixed choice, separated choice or asynchronous)
\cite{
Boreale98,
Nestmann00,
NestmannP00,
Parrow00,
CardelliG00,
CardelliGG02,
BusiGZ03,
BusiGZ09,
CarboneM03,
Palamidessi03,
BPV04,
BPV05,
Palamidessi05,
Nestmann06,
PalamidessiSVV06,
PV06,
CCP07,
VPP07,
CCAV08,
HMP08,
Parrow08,
PV08,
VBG09,
PSN11,
PN12}; see
\cite{Gorla10b,Gorla10a} for an overview. Many of these results employ different and
somewhat ad-hoc criteria on what constitutes a valid encoding, and thus are hard to
compare with each other. Gorla \cite{Gorla10a} collected some essential features of
these approaches and integrated them in a proposal for a valid encoding that justifies
most encodings and some separation results from the literature.

Like Boudol \cite{Bo85} and the present paper, Gorla requires a compositionality condition
for encodings. However, his criterion is weaker than mine (cf.\ \df{compositionality}) in
that the expression $E_f$ encoding an operator $f$ may be dependent on the set of \emph{names}
occurring freely in the expressions given as arguments of $f$.
The reason for this weakening appears to be that it provides a method
for freeing up names that need to be fresh because of the special r\^ole they play in the
translation, but might otherwise occur in the expressions being translated.

To address the problem of freeing up names I advocate a slightly different approach,
already illustrated in \sect{CSP-CCS}:
Most languages with names are parametrised with the set of names that are allowed in
expressions. So instead of the single language CCS, there is an incarnation
CCS($\mathcal{A}$) for each choice of names $\mathcal{A}$. Likewise, there is an incarnation
CSP($\mathcal{A}$) of CSP for each $\mathcal{A}$. A priori, these parameters need not be
related. So rather than insisting that for every $\mathcal{A}$ the language
CCS($\mathcal{A}$) encodes CSP($\mathcal{A}$), I merely require that for each $\mathcal{A}$
there exists a $\mathcal{B}$ such that CCS($\mathcal{B}$) encodes CSP($\mathcal{A}$).
Now the translations obviously are also parametrised by the choice of $\mathcal{A}$, and
they may use names in $\mathcal{B}-\mathcal{A}$ as names that are guaranteed to be fresh.
It is an interesting topic for future research to see if there are any valid encodability
results \`a la \cite{Gorla10a} that suffer from my proposed strengthening of compositionality.

The second criterion of \cite{Gorla10a} is a form of invariance under name-substitution.
It serves to partially undo the effect of making the compositionality requirement
name-dependent. In my setting I have not yet found the need for such a condition.
This criterion as formalised in \cite{Gorla10a} is too restrictive. It forbids the
translation of the input process $a(x).E$ from value-passing CCS \cite{Mi90ccs} into the
CCS expression $\sum_{v\in \cal V}a_v.E[v/x]$, where $\cal V$ is a given (possibly
infinite) set of data values. The problem is that a renaming of the single name $a$
occurring in an expression $E$ of value-passing CCS, say into $b$, would require renaming
infinitely many names $a_v$ occurring in $\fT(E)$ into $b_v$, which is forbidden in \cite{Gorla10a}.
Yet this translation, from \cite{Mi90ccs}, appears entirely justified intuitively.

The remaining three requirements of Gorla might be seen as singling our a particular
preorder $\sqsubseteq$ for comparing terms and their translations. Since in
\cite{Gorla10a}, as in \cite{Bo85}, the domain
of interpretation consists of the closed expressions, and $\sqsubseteq$ is generally not a
congruence for the source or target languages, one needs to compare with the approach of
\sect{respects}, where $\sim$ is allowed to be a preorder.
The preorder presupposes a transition system with $\tau$-transitions (reduction), and a
notion of a success state; and compares processes based on these attributes only.

Hence Gorla's criteria are very close to an instantiation of mine with a particular
preorder. Further work is needed to sort out to what extent the two approaches have
relevant differences when evaluating encoding and separation results from the literature.
Another topic for future work is to sort out how dependent known encoding and separation
results are on the chosen equivalence or preorder. 

As a concluding remark, many separation results in the
literature\cite{CarboneM03,Palamidessi03,PalamidessiSVV06,PV06,PV08,HMP08} are
based on the assumption that parallel composition translates homomorphically,
i.e.\ $\fT(E|F)=\fT(E)|\fT(F)$.\footnote{This assumption is often defended by the theory that
non-homomorphic translations reduce the degree of concurrency of the source process---a
theory I do not share. Note that my translation of CSP into CCS in \sect{CSP-CCS} is not homomorphic.}
This applies for instance to the proof in \cite{HMP08} that there is no valid encoding from
the asynchronous $\pi$-calculus into CCS\@. In \cite{Gorla10a} this assumption is relaxed,
but the separation proof of \cite{Gorla10a} hinges crucially on the too restrictive form
of Gorla's second criterion. 
Whether the asynchronous $\pi$-calculus is expressible in CCS is therefore still wide open.
\vspace{-1ex}

\paragraph{Acknowledgement}
My thanks to an EXPRESS/SOS referee for careful proofreading.

\bibliographystyle{eptcs}
\bibliography{$HOME/Stanford/lib/abbreviations,$HOME/Stanford/lib/dbase,$HOME/Stanford/lib/biblio/glabbeek,$HOME/Stanford/lib/new}
\end{document}